\documentclass[12pt,a4paper]{article}
\usepackage{fullpage}
\usepackage{amssymb,amsmath,stmaryrd,amsthm}
\usepackage[mathscr]{eucal}
\usepackage{pstricks,pst-node,pst-text,pst-3d}
\usepackage[round]{natbib}

\theoremstyle{definition}
\newtheorem{definition}{Definition}
\newtheorem{theorem}{Theorem}

\newtheorem{lemma}{Lemma}

\newtheorem{example}{Example}
\theoremstyle{remark}

\newcommand{\refe}[1] {(\ref{#1})}
\def\dep {\hspace*{-\parindent}}

\def \RR {\mathbb{R}}
\def \NN {\mathbb{N}}
\def \ZZ {\mathbb{Z}}

\def \cP {\mathcal{P}}
\def \cS {\mathcal{S}}
\def \supp {\mathrm{supp}}
\def \0{{\mathbf{0}}}

\def \zero {\mathbb{O}}  

\newcommand{\pp}[1] { \left( {#1} \right)  }

\newcommand{\pr}[1] { \left[ {#1} \right]  }

\begin{document}

\title{On the decomposition of Generalized Additive Independence models}
\author{Michel GRABISCH${}^1$\thanks{Corresponding author.} and Christophe LABREUCHE${}^2$\\
\normalsize ${}^1$ Paris School of Economics, University of Paris I\\
\normalsize 106-112, Bd de l'H\^opital, 75013 Paris, France\\
\normalsize \tt michel.grabisch@univ-paris1.fr\\
\normalsize ${}^2$ Thales Research and Technology\\
\normalsize 1, Avenue Augustin Fresnel, 91767 Palaiseau, France\\
\normalsize \tt christophe.labreuche@thalesgroup.com}

\date{}
\maketitle

\begin{abstract}
The GAI (Generalized Additive Independence) model proposed by Fishburn is a
generalization of the additive utility model, which need not satisfy mutual
preferential independence. Its great generality makes however its application and study
difficult. We consider a significant subclass of GAI models, namely the discrete
2-additive GAI models, and provide for this class a decomposition into
nonnegative monotone terms. This decomposition allows a reduction from
exponential to quadratic complexity in any optimization problem involving
discrete 2-additive models, making them usable in practice.
\end{abstract}

\section{Introduction}
The theory of multiattribute utility (MAUT) provides an adequate and widely
studied framework for the representation of preferences in decision making with
multiple objectives or criteria (let us mention here only the classic works of
\cite{kera76}, and \cite{krlusutv71} on conjoint measurement, among numerous
other ones). The most representative models in MAUT are the additive utility
model $U(x)=\sum_i u_i(x_i)$, and the multiplicative model (see
\cite{dysa79}), whose characteristic property is the (mutual) preferential
independence, stipulating that the preference among two alternatives should not
depend on the attributes where the two alternatives agree (see \cite{absu15} for
a detailed study on MAUT models satisfying preferential independence).

However, it is well known that in real situations, preferential independence
could be easily violated, because of the possible interaction between
objective/criteria. Referring to the example of evaluation of students in
\cite{gra95a} where students are evaluated on three subjects like mathematics,
physics and language skills, the preference between two students may be inverted
depending on their level in mathematics, assuming that the evaluation policy pays
attention to scientific subjects. For instance, the following preference
reversal is not unlikely (marks are given on a 0-100 scale, in the following
order: mathematics, physics and language skills): $(40,90,60)\succ (40,60,90)$
and $(80,90,60)\prec(80,60,90)$, because if a student is weak in one of the
scientific subject (e.g., 40 in mathematics), more attention is paid to the
other scientific subject (here, physics), otherwise more attention is paid to
language skills.

\medskip

To escape preferential independence, \cite{krlusutv71} have proposed the
so-called \textit{decomposable model}, of the form $U(x)=F(u_1(x_1),\ldots,
u_n(x_n))$, where $F$ is strictly monotone. This model, which is a
generalization of the additive utility model, is characterized by a much weaker
property than preferential independence, namely \textit{weak independence} or
\textit{weak separability} (\cite{wak89}). This property amounts to
requiring preferential independence only for one attribute versus the others,
and is generally satisfied in practice. Taking $F$ as the Choquet integral
w.r.t. a capacity (\cite{cho53}) permits to have a versatile model,
which has been well studied and applied in practice (see a survey in
\cite{grla07b}). The drawback of these models is that in general they require
commensurate utility functions, i.e., one should be able to compare $u_i(x_i)$
with $u_j(x_j)$ for every distinct $i,j$.

Another generalization of the additive utility model escaping preferential
independence has been proposed by \cite{fis67}, under the name of
\textit{generalized additive independence (GAI)} model. It has the general form
$U(x) = \sum_{S\in \mathcal {S}}u_S(x_S)$, where $\cS$ is any collection of
  subsets of attributes, and $x_S$ is the vector of components of $x$ belonging
  to $S$. This model is very general (it even need not satisfy weak
  independence, see below for an example) and does not need commensurate attributes. 

Its great generality is also the Achille's heel of this model, making it
difficult to use in practice, and so far it has not been so much considered in
the MAUT community. Some developments, essentially focused on the identification
of the parameters of the model, have been done in the field of artificial
intelligence (see, e.g., \cite{bagr95,bobabr01,bifameza12}).  There are two
major difficulties related to this model.

Firstly, its expression is far from being unique. In two equivalent
decompositions $U(x) = \sum_{S\in \mathcal {S}}u_S(x_S) = \sum_{S\in \mathcal
  {S}}u'_S(x_S)$, the utility functions $u_S$ and $u'_S$ may behave completely
differently and in particular be governed by different monotonicity conditions.
This implies that there is no intrinsic semantics attached to these utility
terms, which makes the model difficult to interpret for the decision maker.

The second difficulty is related to its elicitation,
because the
number of monotonicity constraints on the parameters of the model grows
exponentially fast in the number of attributes.  As these constraints must be
enforced, the practical identification of the model appears to be rapidly
computationally intractable as the number of attributes and the cardinality of
the attributes grow.

\medskip

The aim of this paper is to provide a first step in making GAI models usable in
practice, by proving a fundamental result on decomposition, in a subclass of GAI
models which is significant for applications. Specifically, we are interested
in GAI models  where, first, the collection
$\cS$ is made only of singletons and pairs, thus limiting the model to a sum of
univariate or bivariate terms, and second, the attributes take discrete
values. We call this particular class \textit{2-additive discrete GAI}
models. In addition, we assume that weak independence holds.


The main result of this paper shows that for a given 2-additive GAI model that
fulfills weak independence, it is always possible to obtain a decomposition into
nonnegative monotone nondecreasing terms.  The result is proved by using an
equivalence between 2-additive discrete GAI models and 2-additive $k$-ary
capacities, and amounts to finding the set of extreme points of the polytope of
2-additive $k$-ary capacities.  Going back to the first difficulty mentioned
earlier, using this decomposition provides a semantics to the utility terms
$u_S$ as they have the same monotonicity as the overall utility $U$.  Secondly,
thanks to this result, it is possible to replace the monotonicity conditions on
$U$ by monotonicity conditions on each term $u_S$, which reduces the number of
monotonicity constraints from exponential to quadratic complexity.  This is of
extreme importance in practice.

\medskip

The paper is organized as follows. Section~\ref{sec:back} introduces the
necessary concepts and notation in multiattribute utility, capacities, $k$-ary
capacities, and GAI models. Section~\ref{sec:equiv} introduces $p$-additive GAI
models, and shows the equivalence with $p$-additive $k$-ary
capacities. Section~\ref{sec:main} explains the complexity problem behind the
identification of 2-additive discrete GAI models, and proves that a
decomposition into nonnegative monotone nondecreasing terms is always possible,
which constitutes the main result of the paper. 

\section{Background}\label{sec:back}

\subsection{Multi-Attribute Utility Theory}
\label{sec:maut}
We consider $n$ attributes $X_1,\ldots,X_n$, letting $N=\{1,\ldots,n\}$ be its
index set. Alternatives are represented by a vector $x=(x_1,\ldots,x_n)$ in
$X=X_1\times\cdots\times X_n$.  We denote by $(x_A,y_{-A}) \in X$ the compound alternative
taking value $x_i$ if $i\in A$ and value $y_i$ otherwise.

One of the leading model in decision theory is Multi-Attribute Utility Theory
\citep{kera76}.  The overall utility $U:X\rightarrow \RR$ representing the
preference relation $\succcurlyeq$ of a decision maker (i.e. $x \succcurlyeq y$
iff $U(x) \geq U(y)$) is then supposed to satisfy \emph{preferential
  independence}, whereby the comparison between two alternatives does not depend
on the attributes having the same value.  Accordingly, $U$ can take the form of
either an additive model $U(x)=\sum_{i\in N} k_i \: u_i(x_i)$, or a  multiplicative
form $1-k\:U(x) = \prod_{i\in N} (1-k\:k_i\: u_i(x_i))$, where $u_i$ is a
marginal utility function over attribute $X_i$.

As we explained in the introduction, preferential independence is quite a strong
condition which is not always met in practice.  A weaker condition is \emph{weak
  independence} where for all $i \in N$, all $x_i,y_i \in X_i$ and all
$z_{-i},t_{-i}\in X_{-i}$
\[ (x_i,z_{-i}) \succcurlyeq (y_i,z_{-i})  \quad \Longleftrightarrow \quad (x_i,t_{-i}) \succcurlyeq (y_i,t_{-i})
\]
($x_i$ is at least as good as $y_i$ \textit{ceteris paribus}). Under this
condition, we can define a preference relation $\succcurlyeq_i$ on a single
attribute $X_i$ as follows: for all $x_i,y_i \in X_i$
\[ x_i \succcurlyeq_i y_i \quad \mbox{ iff } \quad (x_i,z_{-i}) \succcurlyeq
(y_i,z_{-i}), 
\]
for some $z_{-i}\in X_{-i}$.  

\subsection{Generalized Additive Independence (GAI) model}

The additive utility model $\sum_{i\in N} u_i(x_i)$ can be easily generalized by considering marginal utility functions over subsets of attributes, with potential overlap between the subsets \citep{fis67,bagr95}:
\begin{equation}\label{eq:gai}
U(x) = \sum_{S\in \cS} u_S(x_S)\qquad (x\in X),
\end{equation}
where $\cS\subseteq 2^N\setminus\{\emptyset\}$.  This model is called the
\emph{Generalized Additive Independence} (GAI) model.  It is characterized by a
condition stating that if two probability distributions $P$ and $Q$ over the
alternatives $X$ have the same marginals over every $S\in \cS$, then the
expected utility of $P$ and $Q$ are equal.
The additive utility model is a particular case of the GAI model when
$\mathcal{S}$ is composed of singletons only.
  
Unlike the additive utility model or the multiplicative model, the GAI model does
not necessarily satisfy weak independence.  In the Artificial Intelligence
community, researchers are interested in the representation of preferences that
may violate weak independence.  A well-known example of such a preference is the
following: consider two attributes $X_1,X_2$ where $X_1$ pertains on the type of
wine and $X_2$ to the type of main course in a restaurant. Then usually, one
prefers `red wine' to `white wine' if the main course is `meat', but `white
wine' is preferred to `red wine' if the main
course is `fish' (the preference over attribute `wine' is conditional on the
value on attribute `main course') \citep{bobabr01}.

In this work, we follow a more traditional view of Decision Theory and assume
that weak independence holds, which is the case in most of the decision
problems.

We make the following two assumptions:
\begin{itemize}
\item {\bf Assumption 1:} Monotonicity:
\[
\forall i\in N, x_i\succcurlyeq_i y_i\Rightarrow U(x)\geq U(y)
\] 
\item {\bf Assumption 2:} Boundaries: each $X_i$ is bounded, in the sense that
  there exist $x_i^\top, x_i^\bot\in X_i$ which are the best and worst elements
  of $X_i$ according to $\succcurlyeq_i$, and
\[
U(x_i^\top,\ldots,x_n^\top) = 1, \quad U(x_i^\bot,\ldots,x_n^\bot) = 0.
\]
\end{itemize}

\subsection{Non-uniqueness of the GAI decomposition}
\label{S2.3}

In the additive utility model, considering two possible decompositions 
$U(x)= \sum_{i\in N} u_i(x_i) = \sum_{i\in N} u'_i(x_i)$, $u_i$ and $u'_i$ are equal up to a constant \citep{fis65},
so that all admissible utility functions satisfy the same monotonicity
(for any two $x_i,y_i \in X_i$, we have $u_i(x_i) \geq u_i(y_i)$ iff $u'_i(x_i) \geq u'_i(y_i)$).

Concerning the GAI model, taking two equivalent decompositions $U(x) =
\sum_{S\in\mathcal{S}} u_S(x_S) = \sum_{S\in\mathcal{S}} u'_S(x_S)$, they are
related by \citep{fis67}
\[ u'_S(x_S) = u_S(x_S) + \sum_{S'\in \mathcal{S} \setminus \{S\},\, S\cap S'\neq \emptyset} f_{S,S'}(x_{S\cap S'}) + c_S
\]
where $f_{S,S'}: X_{S\cap S'} \rightarrow \RR$, and $ \sum_{S\in \mathcal{S}}
\pr{ \sum_{S'\in \mathcal{S} \setminus \{S\},\, S\cap S'\neq \emptyset}
  f_{S,S'}(x_{S\cap S'}) + c_S } = 0 $.  Due to the presence of functions
$f_{S,S'}$, we do not have $u_S(x_S) \geq u_S(y_S)$ iff $u'_S(x_S) \geq
u'_S(y_S)$, for any two $x_S,y_S \in X_S$ \cite[page 87]{bra12}.
Moreover, even if $U$ satisfies weak independence, it might be the case that
$u_S$ does not fulfill this condition, or satisfies it but does not have the
same monotonicity as $U$.  Hence {\it there is no well-defined semantics of the
  utility functions $u_S$}, contrarily to what is claimed in \cite[section
  3.2.1.4]{bra12}.

Braziunas proposes a decomposition based on the Fishburn representation. Fixing
an order on $\mathcal{S}$, say, $\mathcal{S} = \{S_1,\ldots, S_p\}$, the overall
utility reads $U(x)=\sum_{S\in \mathcal{S}} u_S^C(x_S)$ with, for every $j\in
\{1,\ldots,p\}$
\begin{equation}
 u^C_{S_j}(x_{S_j}) = U(x[S_j]) + \sum_{K \subseteq \{1,\ldots,j-1\} \,,\: K \not= \emptyset} (-1)^{|K|}
    U\pp{ x\pr{ \cap_{k\in K} S_k \cap S_j }}
\label{Ebraz1}
\end{equation}
where $\cdot^C$ stands for ``canonical'', $\zero\in X$ is any element in $X$ seen as an
anchor, and $x[S] \in X$ defined by $(x[S])_i=x_i$ if $i\in S$ and
$(x[S])_i=\zero_i$ otherwise \cite[page 94]{bra12}).  Note that the
expression depends on the chosen ordering of the elements of $\mathcal{S}$.

\begin{example}
Consider the following function $U(x_1,x_2,x_3)=x_2+x_1\: x_3 + \max(x_1,x_2)$.
We have $\mathcal{S}=\{S_1,S_2,S_3\}$ with $S_1=\{2\}$, $S_2=\{1,3\}$ and $S_3=\{1,2\}$.
Then the canonical decomposition gives, with $\zero=(0,0,0)$:
\begin{align*}
 & u_{S_1}^C(x_2) = U(x[S_1])=U(\zero_1,x_2,\zero_3)=2\: x_2 \\
 & u_{S_2}^C(x_1,x_3)=U(x[S_2])-U(x[S_1\cap S_2])=U(x_1,\zero_2,x_3)-U(\zero)=x_1\: (x_3+1) \\
 & u_{S_3}^C(x_1,x_2) = U(x[S_3])-U(x[S_1\cap S_3])-U(x[S_2\cap S_3])+U(x[S_1 \cap S_2 \cap S_3]) \\
  & \quad   = U(x_1,x_2,\zero_3) - U(\zero_1,x_2,\zero_3) - U(x_1,\zero_2,\zero_3) + U(\zero) \\
	& \quad 	=\max(x_1,x_2)-x_1-x_2=-\min(x_1,x_2)
\end{align*}
We note that $U$ is nondecreasing in all variables, even though, for the canonical decomposition, $u_{S_3}^C$ is nonincreasing in its two coordinates.
\label{Ex1}
\end{example}

\subsection{Capacities and k-ary capacities}
We consider a finite set $N=\{1,\ldots,n\}$ (e.g., the index set of attributes
as in Section~\ref{sec:maut}). A \textit{game} on $N$ is a set function
$v:2^N\rightarrow \RR$ vanishing on the empty set. A game $v$ is
\textit{monotone} if $v(S)\leq v(T)$ whenever $S\subseteq T$. Note that monotone
games take nonnegative values, and if in addition $v(N)=1$, the game is said to
be \textit{normalized}. In the sequel, we will mainly deal with monotone
normalized games, which are usually called \textit{capacities}
\citep{cho53}\footnote{Often capacities are defined as monotone games, not
  necessarily normalized.}.

\medskip

Making the identification of sets with their characteristic functions, i.e.,
$S\leftrightarrow 1_S$ for any $S\in 2^N$, with $1_S:N\rightarrow \{0,1\}$,
$1_S(i)=1$ iff $i\in S$, games can be seen as functions on the set of binary
functions. A natural generalization is then to consider functions taking values
in $\{0,1,\ldots,k\}$, leading to the so-called multichoice or $k$-choice games
\citep{hsra90} and $k$-ary capacities \citep{grla03b}. Formally, a
\textit{$k$-choice game} is a mapping $v:\{0,1,\ldots,k\}^N\rightarrow \RR$
satisfying $v(0,\ldots,0)=0$. A \textit{$k$-ary capacity} is a $k$-choice game
being monotone and normalized, i.e., satisfying $v(y)\leq v(z)$ whenever $y\leq
z$, and $v(k,\ldots, k)=1$.

\medskip

Let $v:2^N\rightarrow \RR$ be a game. The \textit{M\"obius transform} of $v$
(a.k.a. \textit{M\"obius inverse}) is the set function $m^v:2^N\rightarrow \RR$
which is the (unique) solution of the linear system
\[
v(S) = \sum_{T\subseteq S}m^v(T) \qquad (S\in 2^N)
\]
(see \cite{rot64}). It is given by
\begin{equation}\label{eq:mob}
m^v(S) = \sum_{T\subseteq S}(-1)^{|S\setminus T|}v(T) \qquad (S\in 2^N).
\end{equation}
A capacity $v$ is said to be \textit{(at most) $p$-additive} for some
$p\in\{1,\ldots,n\}$ if its M\"obius transform vanishes for subsets of more than
$p$ elements: $m^v(S)=0$ for all $S\subseteq N$ such that $|S|>p$.

Similarly, given a $k$-ary game $v$, its M\"obius transform is defined as
the unique solution of the linear system
\begin{equation}\label{eq:zeta}
v(z) = \sum_{y\leq z}m^v(y) \qquad (z\in \{0,1,\ldots,k\}^N).
\end{equation}
It is shown in the appendix that its solution is given by
\begin{equation}\label{eq:mobk}
m^v(z) = \sum_{y\leq z \ : \ z_i-y_i\leq 1\forall i\in N}(-1)^{\sum_{i\in
    N}(z_i-y_i)}v(y) \qquad (z\in\{0,1,\ldots,k\}^N).
\end{equation}
It follows that any $k$-ary game $v$ can be written as:
\[
v = \sum_{x\in L^N}m^v(x)u_x,
\]
with $u_x$ a $k$-ary capacity defined by
\[
u_x(z) = \begin{cases}
  1, & \text{if } z\geq x\\
  0, & \text{otherwise.}
  \end{cases}
\]
By analogy with classical games, $u_x$ is called the \textit{unanimity game}
centered on $x$. Note that this decomposition is unique as the unanimity games
are linearly independent, and form a basis of the vector space of $k$-ary games.

\section{Relation between GAI and $k$-ary capacities}\label{sec:equiv}

\subsection{Discrete GAI models are $k$-ary capacities}
We consider discrete GAI models, i.e., where attributes can take only a finite
number of values, and show that they are particular instances of $k$-ary
capacities. We put
\[
X_i=\{a_i^0,\ldots,a_i^{m_i}\} \qquad (i\in N),
\]
with $a_i^0\preccurlyeq_i\cdots\preccurlyeq_i a_i^{m_i}$. Any alternative $x\in
X$ is mapped to $\{0,\ldots,m_1\}\times\cdots\times\{0,\ldots,m_n\}$ by the
mapping $\varphi$ which simply keeps track of the rank of the value of the
attribute:
\[
(a_1^{j_1},\ldots, a_n^{j_n}) \mapsto \varphi(a_1^{j_1},\ldots, a_n^{j_n})   =
(j_1,\ldots, j_n).
\]
We consider now the smallest (discrete) hypercube $\{0,\ldots,k\}^N$ containing
$\{0,\ldots,m_1\}\times\cdots\times\{0,\ldots,m_n\}$, with $k:=\max_im_i$. Given
a GAI model $U$ with discrete attributes as described above, we define the
mapping $v:\{0,\ldots,k\}^N\rightarrow \RR$ by
\[
U(x) =: v(\varphi(x)) \qquad (x\in X)
\]
and let $v(z):=v(m_1,\ldots, m_n)$ when $z\in \{0,\ldots,k\}^N\setminus\varphi(X)$.
In words, $v$ encodes the values of $U$ for every alternative, and fills in the
missing values in the hypercube by the maximum of $U$. By assumption 1 and 2 on
$U$, it follows that $v$ is a normalized $k$-ary capacity on $N$. 

From now on, we put $L=\{0,1,\ldots, k\}$.

\subsection{$p$-additive GAI models}
Consider a GAI model $U$ on $X$, where the attributes need not be discrete. As
$U$ is in general exponentially complex in the number of attributes, one is
looking for  simple particular cases. The simplest case would be to consider a
classical additive model. The characteristic property of an additive model is
that the variation of $U$ in one attribute is unrelated to the value of the
other fixed ones:
\[
U(y_i,x_{-i}) - U(x_i,x_{-i}) = u_{\{i\}}(y_i) - u_{\{i\}}(x_i).
\]
Calling the left member the (1st order) variation of $U$ w.r.t. $i$ from $x_i$
to $y_i$ at $x$, we define inductively the \textit{variation of $U$
  w.r.t. $P\subseteq N$ from $x_P$ to $y_P$ at $x$} by
\[
\Delta_{x_P}^{y_P} U(x)  =\sum_{T\subseteq P} (-1)^{|P\setminus
  T|}U(y_T,x_{P\setminus T},x_{-P}) 
\] 
For example, one has, abbreviating $\{i,j\}$ by $ij$:
\begin{align*}
\Delta_{x_i}^{y_i}U(x)  &= U(y_i,x_{-i}) - U(x_i,x_{-i}) \\
\Delta_{x_{ij}}^{y_{ij}}U(x) &= U(y_{ij}, x_{-ij}) - U(x_i,y_j,x_{-ij}) -
U(y_i,x_j,x_{-ij}) + U(x).
\end{align*}
\begin{definition}
A function $U$ on $X$ is said to be \textit{$p$-additive} for some
$p\in\{1,\ldots, n\}$ if for every $P\subseteq N$ with $|P|\leq p$, for every
$x\in X$, $x_P,y_P\in X_P$ and $x'_{-P}\in X_{-P}$, 
\[
\Delta_{x_P}^{y_P}U(x_P,x_{-P}) =\Delta_{x_P}^{y_P}U(x_P,x'_{-P}). 
\]
\end{definition}
The above definition generalizes the notion of 2-additivity proposed in
\cite{lagr13}.

The next theorem relates $p$-additivity to the decomposition of $U$ into terms
involving at most $p$ variables, and generalizes \cite[Prop. 4]{lagr13}.
\begin{theorem}\label{th:1}
A function $U$ on $X$ is $p$-additive for some $p\in\{1,\ldots, n\}$ if and only
if  there exist functions $u_A:X_A\rightarrow \RR$, for every $A\subseteq N$ with $|A|\leq
p$, such that $U$ takes the form (\ref{eq:gai}) with $\cS=\{A\subseteq N,
0<|A|\leq p\}$. 
\end{theorem}
\begin{proof}
We suppose $p\neq n$ to discard the trivial case.
The ``if'' part is easy to check. As for the ``only if'' part, fix $x\in X$ and
define $v(A) = U(x_A,0_{-A})$ for all $A\subseteq N$. By assumptions 1 and 2,
$v$ is a (nonnormalized) capacity on $N$. Define its discrete derivative
inductively as follows. For any $\emptyset\neq S\subset N$, $T\in 2^N$ and
$i\not\in S$,
\[
\Delta_{S\cup i}v(T) = \Delta_i(\Delta_Sv(T))
\]
with $\Delta_iv(T) = v(T\cup i) - v(T)$. Then it is easy to see by
(\ref{eq:mob}) that
$\Delta_Sv(\emptyset)=m^v(S)$, and that for disjoint $S$ and $T$
\[
\Delta_Sv(T) = \Delta_{0_S}^{x_S}U(x_T,0_{-T}).
\]
Take $S$ such that $|S|=p$ and any $i\in N\setminus S$. Then for any $T\subseteq
N\setminus (S\cup i)$,
\[
\Delta_{S\cup i}v(T) = \Delta_i(\Delta_Sv(T)) = \Delta_{0_S}^{x_S}U(x_{T\cup
  i},0_{-T\cup i}) - \Delta_{0_S}^{x_S}U(x_T,0_{-T}) = 0
\]
by assumption of $p$-additivity of $U$. Letting $T=\emptyset$, it follows that
$v$ is $p$-additive too (in the sense of capacities), hence we can write:
\[
U(x) = v(N) = \sum_{S\subseteq N, 0<|S|\leq p}m^v(S)
\]
with $m^v(S) = \Delta_Sv(\emptyset) = \Delta_{0_S}^{x_S}U(\0)$. Since the latter
term only depends on the variables $x_S$, the desired result follows.
\end{proof}

\subsection{$p$-additive $k$-ary capacities}
By analogy with classical capacities, a $k$-ary capacity $v$ is said to be
\textit{(at most) $p$-additive} if $m^v(z)=0$ whenever
  $|\supp(z)|>p$, where
\[
\supp(z) = \{i\in N\mid z_i>0\}.
\]
\begin{lemma}\label{lem:1}
Let $k\in \NN$ and $p\in\{1,\ldots,n\}$. A $k$-ary game $v$ is
$p$-additive if and only if it has the form
\begin{equation}\label{eq:padd}
v(z) = \sum_{x\in L^N, 0<|\supp(x)|\leq p}v_x(x\wedge z) \qquad (z\in L^N)
\end{equation}
where $v_x:L^N\rightarrow \RR$ with $v_x(\0)=0$. 
\end{lemma}
\begin{proof}
Suppose that $v$ is $p$-additive. By the decomposition of $v$ in the basis of
unanimity games, it follows that
\[
v = \sum_{x\in L^N, 0<|\supp(x)|\leq p}m^v(x)u_x,
\]
hence we have the required form with $v_x=m^v(x)u_x$. Conversely, again by
decomposition in the basis of unanimity games and since $v_x$ is a game,
(\ref{eq:padd}) can be rewritten as:
\begin{align*}
\sum_{y\in L^N}m^v(y)u_y(z) & =  \sum_{x\in L^N, 0<|\supp(x)|\leq p}\sum_{y\in
  L^N}m^{v_x}(y)u_y(x\wedge z) \\
 & = \sum_{y\in L^N, 0<|\supp(y)|\leq p}\sum_{x\in L^N,0<|\supp(x)|\leq
  p}m^{v_x}(y)u_y(x\wedge z)\\
 & = \sum_{y\in L^N, 0<|\supp(y)|\leq p}\Big(\sum_{x\geq y,0<|\supp(x)|\leq p}m^{v_x}(y)\Big)u_y(z).
\end{align*}
By uniqueness of the decomposition, it follows that $v$ is $p$-additive.
\end{proof}
Note that even if $v$ is a capacity, the $v_x$ are not necessarily capacities.

It follows from Theorem~\ref{th:1} and the above result that the set of
$p$-additive discrete GAI models on $X$ coincides with the set of (at most)
$p$-additive $k$-ary capacities.

\section{Monotone decomposition of a 2-additive GAI model}\label{sec:main}
\subsection{A complexity problem}
 We have seen in Section \ref{S2.3} that the GAI decomposition is not unique.
 Moreover, the terms in two equivalent GAI decompositions may have different
 monotonicity conditions, as shown in Example \ref{Ex1}.  Then the following
 question arises: \textit{Given a GAI model, is it always possible to get a
   decomposition into nonnegative nondecreasing terms?}  The main result of this
 paper will give a positive answer to this question, in the case of 2-additive
 GAI models. This case is of particular importance in practice, since it
 constitutes a good compromise between versatility and complexity. Experimental
 studies in multicriteria evaluation have shown that 2-additive capacities have
 almost the same approximation ability than general capacities (see, e.g.,
 \cite{grdulipe01}).  A two-additive GAI model is considered in
 \cite{bifameza12}, and a very similar model is defined in \cite{grmosl12}.

Before stating and proving the
result, we explain why it is important to solve this problem, which is related
to the complexity of the model.

We begin by computing the number of unknowns in a 2-additive GAI model
equivalent to a $k$-ary capacity. By Theorem~\ref{th:1}, such a model has the
form (\ref{eq:gai}) with $\cS$ being the set of singletons and pairs. Since
$|L|=k+1$, this yields
\[
(k+1)\binom{n}{1} + (k+1)^2\binom{n}{2} = \frac{n(k+1)}{2}\Big(2+(k+1)(n-1)\Big)
\]
unknowns. $U$ being monotone nondecreasing, this induces a number of
monotonicity constraints on the unknowns, of the type
\begin{equation}\label{eq:mono}
U(a_1^{j_1},\ldots,a_{i-1}^{j_{i-1}},a_i^{j_i+1},a_{i+1}^{j_{i+1}},\ldots,a_n^{j_n})
\geq U(a_1^{j_1},\ldots,a_{i-1}^{j_{i-1}},a_i^{j_i},a_{i+1}^{j_{i+1}},\ldots,a_n^{j_n})
\end{equation}
for every $i\in N$,
$j_1\in\{0,\ldots,m_1\},\ldots,j_{i-1}\in\{0,\ldots,m_{i-1}\},j_i\in\{0,\ldots,m_i-1\}$,
$j_{i+1}\in\{0,\ldots,m_{i+1}\}$,\ldots,$j_n\in\{0,\ldots,m_n\}$. The number of
elementary conditions contained in (\ref{eq:mono}) is equal to
\[
\sum_{i\in N}\Big(m_i\times\prod_{j\in N\setminus\{i\}}(m_j+1)\Big).
\]
In the case where $m_i=k$ for every $i$, this number becomes
\[
n\times k\times (k+1)^{n-1}.
\]
Although the number of variables was still quadratic in $n$ and $k$, the number
of constraints is exponential in $n$. It follows that any practical identification
  of a GAI model based on some optimization procedure\footnote{The learning
    problem can be classically transformed into a linear program, where the
    training set is seen as linear constraints on the GAI variables \citep{bifameza12,grmosl12}. It could also be possible to perform statistical learning, like in \cite{fachdehu12}, where the underlying optimization problem is a convex problem under linear constraints.}, where the variables are
  the unknowns of the GAI model and the constraints are the monotonicity
  constraints (\ref{eq:mono}) plus possibly some learning data,  has to cope with an
  exponential number of constraints. 
The following tables, obtained with $k=4$,
shows that the underlying optimization problem becomes rapidly intractable.
\begin{center}
\begin{tabular}{|p{3cm}|c|c|c|c|}\hline
$n$ & 4 & 6 & 8 & 10 \\ \hline
 $\sharp$ of variables  & 170 & 405 & 740 &  1175\\ \hline
$\sharp$ of constraints & 2000 & 75 000 & 2 500 000 & 78 125 000  \\ \hline
\end{tabular}
\medskip
\begin{tabular}{|p{3cm}|c|c|c|}\hline
$n$  & 12 & 14 & 20 \\ \hline
$\sharp$ of variables  & 1710 & 2345 & 4850 \\ \hline
$\sharp$ of constraints  & 2 343 750 000 & 68 359 375 000 &
  $1.526E+15$ \\ \hline
\end{tabular}
\end{center}

However, if a decomposition into nonnegative nondecreasing terms is possible, one
  has only to check monotonicity of each term. Then the number of monotonicity
  conditions drops to 
\[
\sum_{i\in N}m_i + \sum_{\{i,j\}\subseteq N}\big(m_i(m_j+1)+m_j(m_i+1)\big).
\]
In the case where $m_i=k$ for every $i$, this number becomes
\[
n\times k\times \Big[(n-1)(k+1) +1\Big],
\]
which is quadratic in $n$. The following table ($k=4$) shows that the
optimization problem becomes tractable even for a large number of attributes.
\begin{center}
\begin{tabular}{|p{4.5cm}|c|c|c|c|c|c|c|}\hline
$n$ & 4 & 6 & 8 & 10 & 12 & 14 & 20\\ \hline
$\sharp$ of constraints with monotone decomposition & 256 & 624 & 1152 & 1840 & 2688 & 3696
  & 7680\\ \hline
\end{tabular}
\end{center}

\subsection{The main result}
The following theorem states that a decomposition of a 2-additive GAI model into
monotone nondecreasing terms is always possible.
\begin{theorem}\label{th:main}
Let us consider a 2-additive discrete GAI model $U$ satisfying assumptions 1 and
2. Then there exist nonnegative and nondecreasing functions $u_i:X_i\rightarrow
[0,1]$, $i\in N$, $u_{ij}:X_i\times X_j\rightarrow [0,1]$, $\{i,j\}\subseteq N$,
such that
\[
U(x) = \sum_{i\in N}u_i(x_i) + \sum_{\{i,j\}\subseteq N}u_{ij}(x_i,x_j) \qquad
(x\in X)
\]
\end{theorem}

The rest of this section is devoted to the proof of this theorem, which goes
through a number of intermediary results. First, we remark that the problem is
equivalent to the decomposition of a 2-additive $k$-ary capacity $v$
into a sum of 2-additive $k$-ary capacities whose support has size at most 2,
where the \textit{support} of $v$ is defined by
\[
\supp(v) = \bigcup_{x\in L^N: m^v(x)\neq 0}\supp(x).
\]

We consider $\cP_{k,\cdot}$ the polytope of $k$-ary capacities, and
$\cP_{k,2}$ the polytope of 2-additive $k$-ary capacities. Our aim is
to study the vertices of the latter, and we will show that these vertices are the
adequate $k$-ary capacities to perform the decomposition. 

A first easy fact is that the extreme points of $\cP_{k,\cdot}$ are the 0-1-valued
$k$-ary capacities.
\begin{lemma}\label{lem:k}
$\hat{v}$
is an extreme point of $\cP_{k,\cdot}$ iff $\hat{v}$ is 0-1-valued.
\end{lemma}
\begin{proof}
Take $\hat{v}$ in $\cP_{k,\cdot}$ which is 0-1-valued, and consider $v,v'\in
\cP_{k,\cdot}$ such that $\frac{v+v'}{2}=\hat{v}$. Then, since
$\hat{v}$ is 0-1-valued,
\[
v(x) + v'(x) = \begin{cases}
  2, & \text{if } \hat{v}(x)=1\\
  0, & \text{otherwise}.
  \end{cases}
\]
Since $v,v'$ are normalized and monotone, the only possibility to get
$v(x)+v'(x)=2$ is to have $v(x)=v'(x)=1$, and similarly, $v(x)+v'(x)=0$ forces
$v(x)=v'(x)=0$. It follows that $v=v'=\hat{v}$, i.e., $\hat{v}$ is an extreme
point of $\cP_{k,\cdot}$. 

 Conversely, consider a vertex $\hat{v}$ which is not 0-1-valued, and let 
\[
\epsilon = \min(1-\max_{x:\hat{v}(x)<1} \hat{v}(x),\min_{x:\hat{v}(x)>0}\hat{v}(x)).
\]
Define
\begin{align*}
v'(x) & = \hat{v}(x) + \epsilon, \text{ for all }x\text{ s.t. } \hat{v}(x)\neq
0,1\\
v''(x) & = \hat{v}(x)-\epsilon, \text{ for all }x\text{ s.t. } \hat{v}(x)\neq
0,1,
\end{align*}
and $v'=v''=\hat{v}$ otherwise. Then $v',v''\in \cP_{k,\cdot}$ and $\hat{v} =
\frac{v'+v''}{2}$, a contradiction.
\end{proof}

\begin{lemma}\label{lem:3}
 Let $k\in \NN$ and $v\in \cP_{k,2}$. Then $v$ is 0-1-valued iff $m^v$ is
 $\{-1,0,1\}$ valued.
\end{lemma}
\begin{proof}
$\Leftarrow)$ By the assumption $\sum_{y\leq x}m^v(y)\in \ZZ$ for every $x\in
  \{0,1,\ldots,k\}^N$. Since $v\in \cP_{k,2}$ it follows that $v$ is 0-1-valued.

$\Rightarrow)$ Assume $v$ is 0-1-valued and use (\ref{eq:mobk}) to compute the
  M\"obius transform. For $z=\ell_i$ with $\ell\in\{1,\ldots, k\}$, we have
  $m^v(z)=v(\ell_i)-v((\ell-1)_i)$, so that the desired result holds. Otherwise
  $z=\ell_i\ell'_j$ with $\ell,\ell'\in\{1,\ldots,k\}$ and distinct $i,j\in
  N$. Then 
\begin{equation}\label{eq:p2}
m^v(z) = v(z)-v((\ell-1)_i\ell'_j) - v(\ell_i(\ell'-1)_j) +
  v((\ell-1)_i(\ell'-1)_j).
\end{equation}
 By the assumption and monotonicity of $v$, it follows that $m^v(z)\in\{-1,0,1\}$.
\end{proof}

We recall that a $m\times n$ matrix is totally unimodular if the determinant of
every square submatrix is equal to $-1$, 0 or 1. A polyhedron is integer if all
its extreme points have integer coordinates. Then a matrix $A$ is totally
unimodular iff the polyhedron $\{x\mid Ax\leq b\}$ is integer for every integer
vector $b$. In particular it is known that the vertex-arc matrix $M$ of a
directed graph, i.e., whose entries are $M_{x,a}=1$ if the arc $a$ leaves vertex
$x$, $-1$ if $a$ enters $x$, and 0 otherwise, is totally unimodular (in other
words, each column of $M$ has exactly one $+1$ and one $-1$, the rest being 0).

We are now in position to characterize the extreme points of $\cP_{k,2}$. 
\begin{theorem}\label{th:2}
Let $k\in \NN$. The set of extreme points of $\cP_{k,2}$, the polytope of
2-additive $k$-ary capacities, is the set of 0-1-valued 2-additive $k$-ary
capacities.
\end{theorem}
\begin{proof}
By Lemma~\ref{lem:k}, we need only to prove that any extreme point of $\cP_{k,2}$
is 0-1-valued.

1. We prove that $A_{k,\cdot}$, the matrix defining the polytope of $k$-ary
capacities, is totally unimodular. The argument follows the one given for
classical capacities by  Miranda et al. \cite[Th. 2]{micogi06}. We prove that
$A^\top_{k,\cdot}$ is totally unimodular, which is equivalent to the desired
result. Since the monotonicity constraints are either of the form $v(1_i)\geq 0$
or $v(x)-v(x')\geq 0$ where $x'$ is a lower neighbor of $x$ (i.e. $x'=x-1_i$ for
some $i$), the matrix $A^\top_{k.\cdot}$ has the form $(I,B)$, where $I$ is a
submatrix of the $(k^n-1)$-dim identity matrix $I_{k^n-1}$, and $B$  is a matrix where each
column has exactly one $+1$ and one $-1$. Hence $B$ is totally unimodular, and
so is $(I_{k^n-1},B)$ as it easy to check. Since $A^\top_{k,\cdot}$ is a
submatrix of it, it is also totally unimodular.

2. It follows from Step 1 that the polytope $\cP_{k,\cdot}(b)$ given by
$A_{k,\cdot}v\leq b$ is integer for every integer vector $b$. Next, consider the
$(k^n-1)\times (k^n-1)$-matrix $Z$ expressing the Zeta transform, i.e.,
$Zm^v=v$, as given by (\ref{eq:zeta}). This matrix has only 0 and 1 as entries,
and its inverse $Z^{-1}$ exists and its entries are $0,-1,+1$ only (see
(\ref{eq:mobk})). Consider the polytope $\cP^m_{k,\cdot}(b)$ given by
$A^m_{k,\cdot}m\leq b$ with $A^m_{k,\cdot}=A_{k,\cdot}Z$, the image by the
linear transform $Z$ of the polytope $\cP_{k,\cdot}(b)$.  It is easy to check
that $\hat{v}$ is an extreme point of $\cP_{k,\cdot}(b)$ iff $Z^{-1}\hat{v}$ is
an extreme point of $\cP^m_{k,\cdot}(b)$. Evidently, the coordinates of
$Z^{-1}\hat{v}$ are integer, therefore $\cP^m_{k,\cdot}(b)$ is integer for every
integer vector $b$. We conclude that $A^m_{k,\cdot}$ is totally unimodular.

3. Inasmuch as a submatrix of a totally unimodular matrix is itself totally
unimodular, it follows from Step~2 that $A^m_{k,2}$, the matrix defining the set
of 2-additive $k$-ary capacities in M\"obius coordinates, is also totally
unimodular. As a conclusion, the extreme points of $\cP^m_{k,2}$ are
integer-valued. 

4. We show that the extreme points of $\cP^m_{k,2}$ are $\{-1,0,1\}$-valued. Then
Lemma~\ref{lem:3} permits to conclude. It suffices to show that $|m^v(z)|\geq 2$
cannot happen. If $z=\ell_i$ with $\ell\in\{1,\ldots,k\}$, we find by
(\ref{eq:mobk}) that $m^v(z)=v(\ell_i)-v((\ell-1)_i)$, so that the claim holds
since $v\in \cP_{k,2}$. Otherwise, $z=\ell_i\ell'_j$ with
$\ell,\ell'\in\{1,\ldots,k\}$ and distinct $i,j$, and $m^v(z)$ is given by
(\ref{eq:p2}). Since $v$ is monotone and normalized, the claim easily follows.
\end{proof}

The last step is to prove that a 0-1-valued 2-additive $k$-ary
capacity has a support of size at most 2.
\begin{theorem}\label{th:support}
Consider a 2-additive $k$-ary capacity $u$  on $N$ which is
0-1-valued. Then the support of $u$ is restricted to at most two attributes.
\end{theorem}
\begin{proof}

{\bf Preliminary Step.} $u$ being 2-additive, its expression is
\begin{equation}
u(x)=\sum_{\{i,j\}\subseteq N} u_{i,j}(x_i,x_j) \qquad (x\in X).
\label{E0-1.2}
\end{equation}
If we set $u'_{i,j}(x_i,x_j)=u_{i,j}(x_i,x_j)-u_{i,j}(0,0)$, we obtain
$u(x)=\sum_{\{i,j\}\subseteq N} u'_{i,j}(x_i,x_j) + C$, where $C=
-\sum_{\{i,j\}\subseteq N} u_{i,j}(0,0)$. By assumption~2 and $u'_{i,j}(0,0)=0$, one gets $C=0$.
This proves that in decomposition \refe{E0-1.2}, one can always assume that
\begin{equation}
  \forall \{i,j\}\subseteq N \qquad u_{i,j}(0,0)=0 .
\label{E0-1.3}
\end{equation}
We wish to prove that $u$ depends only on one term $u_{i,j}$ in \refe{E0-1.2}.
In order to avoid cases where such a term $u_{i,j}$ depends only on one variable (in which case $u$ might also depend on another term $u_{k,l}$), we are interested in terms $u_{i,j}$ depending on its two variables $x_i$ and $x_j$.
We say that $u_{i,j}$ \emph{depends on its two variables} if
\begin{align}
 & \exists y_i \in X_i \ \exists y_j \in X_j \qquad u_{i,j}(y_i,y_j) \not= u_{i,j}(y_i,0)  \label{E0-1.4} \\
 & \exists y'_i \in X_i \ \exists y'_j \in X_j \qquad u_{i,j}(y'_i,y'_j) \not= u_{i,j}(0,y'_j)  \label{E0-1.5} 
\end{align}
Clearly, if \refe{E0-1.4} (resp. \refe{E0-1.5}) is not fulfilled, then $u_{i,j}$ does not depend on attribute $x_j$ (resp. $x_i$).

The proof is organized as follows.  We show in Step~1 that if there is no term
$u_{i,j}$ that depends on its two variables, then $u$ depends only on one
variable.  We then assume that at least one term $u_{i,j}$ depends on its two
variables -- denoted $u_{1,2}$ w.l.o.g.  Step~2 shows that it is not possible to
have a non-zero term $u_{i,j}$, with $\{i,j\} \subseteq N\setminus \{1,2\}$.
Step~3 proves that it is not possible to have a non-zero term $u_{i,j}$, with
$i\in \{1,2\}$ and $j \in N\setminus \{1,2\}$.  We conclude that $u_{1,2}$ is
the only non-zero term in the decomposition.  This proves that $u$ depends only
on two variables.

\medskip

{\bf Step 1: case of the additive utility model.} We first start with the case where there is no term $u_{i,j}$ that depends on its two variables.

\begin{lemma}
Assume that there is no term $u_{i,j}$ that depends on its two variables.
Then the support of $u$ is restricted to one attribute.
\label{L0-1.1}
\end{lemma}
\begin{proof}
If there is no term $u_{i,j}$ that depends on its two variables, $u$ takes the form of an additive utility:
\[ u(x) = \sum_{i\in N} u_i(x_i) 
\]
where $u_i:X_i \rightarrow \RR$ is not necessarily nonnegative or monotone.
By \refe{E0-1.3}, we have $u_i(0)=0$ for every $i\in N$.

Let $i\in N$, we write
$u(x_i,0_{-i}) = u_i(x_i)$.
Hence $u_i$ is 0-1-valued and monotone.

As $u$ is not constant by Assumption 2, at least one term $u_i$ is not constant.
W.l.o.g. let us assume it is $u_1$.
Then there exists $x_1\in X_1$ such that $u_1(x_1)=1$.

Now for every $i\in N\setminus \{1\}$ and $x_i\in X_i$,
$u(x_1,x_i,0_{-1,i}) = 1 + u_i(x_i)$.
As $u_i$ is nonnegative and $u$ is 0-1-valued, we conclude that $u_i(x_i)=0$.
Hence $u$ depends only on $x_1$.
\end{proof}

\medskip

{\bf Step 2: Case where $u$ has two non-zero terms with non-overlapping support,
  e.g., $u_{1,2}$ and $u_{3,4}$.} We now focus on the situation where at least
one term $u_{i,j}$ depends on its two variables.  W.l.o.g., we assume it is
$u_{1,2}$.

We consider the general case where there are at least $4$ attributes.
The restriction with only $3$ attributes will be handled in Step 3. For every $j\in N\setminus \{1,2\}$, we choose $k(j) \in N\setminus \{1,2,j\}$
(where $k(j) \not= k(j')$ for $j\not=j'$).
For every $i\in \{1,2\}$ and $j\in N\setminus \{1,2\}$, we set
\begin{align*}
 & u'_{i,j}(x_i,x_j) = u_{i,j}(x_i,x_j) - u_{i,j}(x_i,0) - u_{i,j}(0,x_j)  \\
 & u'_{1,2}(x_1,x_2)=u_{1,2}(x_1,x_2) + \sum_{j\in N\setminus \{1,2\}} (u_{1,j}(x_1,0) + u_{2,j}(x_2,0) ) \\
 & u'_{j,k(j)}(x_j,x_{k(j)}) = u_{j,k(j)}(x_j,x_{k(j)}) + u_{1,j}(0,x_j) + u_{2,j}(0,x_j)
\end{align*}
Then $u(x)=\sum_{\{i,j\}\subseteq N} u'_{i,j}(x_i,x_j)$.
Moreover $u'_{i,j}(x_i,0)=0$ and $u'_{i,j}(0,x_j)=0$ for $i\in \{1,2\}$, $j\in N\setminus \{1,2\}$, $x_i\in X_i$ and $x_j\in X_j$.
Hence in decomposition \refe{E0-1.2}, we can assume that
\begin{equation}
  \forall i\in \{1,2\} \ \forall j\in N\setminus \{1,2\} \ \forall x_i\in X_i \ \forall x_j\in X_j
	 \qquad  u_{i,j}(x_i,0)=0 \mbox{ and } u_{i,j}(0,x_j)=0 .
\label{E0-1.6}
\end{equation}

Thanks to \refe{E0-1.3} and \refe{E0-1.6}, we have
\begin{equation} 
  u(x_1,x_2,0_{-1,2}) = u_{1,2}(x_1,x_2)
\label{E0-1.7}
\end{equation}
Hence
\begin{equation} 
  u_{1,2} \mbox{ is 0-1-valued and monotone} .
\label{E0-1.8}
\end{equation}
By \refe{E0-1.8}, conditions \refe{E0-1.4} and \refe{E0-1.5} with $i=1$, $j=2$ give
\begin{equation} 
  \begin{array}{l}
	  \displaystyle u_{1,2}(y_1,y_2)=1 \ , \ u_{1,2}(y_1,0)=0 \\
		\displaystyle u_{1,2}(y'_1,y'_2)=1 \ , \ u_{1,2}(0,y'_2)=0
	\end{array}
\label{E0-1.9}
\end{equation}

Assume by contradiction that there exists a non-zero $u_{i,j}$ for some $\{i,j\} \subseteq N\setminus\{1,2\}$. 
W.l.o.g., we assume it is $u_{3,4}$.
Then there exists $z_3 \in X_3$ and $z_4 \in X_4$ such that $u_{3,4}(z_3,z_4) \not=0$.
As for \refe{E0-1.6}, we can transfer, for $i\in \{3,4\}$ and $j\in N\setminus\{1,2,3,4\}$,
the term $u_{i,j}(x_i,0)$ in $u_{3,4}$. Hence we can assume that
\begin{equation}
  \forall i\in \{3,4\} \ \forall j\in N\setminus\{1,2,3,4\} \ \forall x_i\in X_i \qquad  u_{i,j}(x_i,0)=0 .
\label{E0-1.10}
\end{equation}
Thanks to \refe{E0-1.3}, \refe{E0-1.6} and \refe{E0-1.10}, we have
\begin{equation} 
  u(x_3,x_4,0_{-3,4}) = u_{3,4}(x_3,x_4)
\label{E0-1.11}
\end{equation}
Hence
\begin{equation} 
  u_{3,4} \mbox{ is 0-1-valued, monotone, and } u_{3,4}(z_3,z_4) = 1 .
\label{E0-1.12}
\end{equation}

\begin{lemma}
If $u_{1,2}$ depends on its two variables, then $u_{3,4}$ is identically zero.
\label{L0-1.2}
\end{lemma}
\begin{proof}
We set $v(x_1,x_2,x_3,x_4) = u(x_1,x_2,x_3,x_4,0_{-1,2,3,4})$. We write
\[ v(x_1,x_2,x_3,x_4) = \sum_{1 \leq i < j \leq 4} u_{i,j}(x_i,x_j) .
\]

\dep\underline{\bf Analysis with $y$ and $z$:}
\begin{itemize}
\item $v(y_1,y_2,z_3,z_4) = \underbrace{u_{1,2}(y_1,y_2)}_{=1} + \underbrace{u_{3,4}(z_3,z_4)}_{=1} 
+ \sum_{i\in\{1,2\},j\in\{3,4\}} u_{i,j}(y_i,z_j)$.
We have $v(y_1,y_2,z_3,z_4)=1$ as $v(y_1,y_2,z_3,z_4) \geq v(y_1,y_2,0,0)=u_{1,2}(y_1,y_2)=1$.
Hence
\begin{equation} 
  \sum_{i\in\{1,2\},j\in\{3,4\}} u_{i,j}(y_i,z_j) = -1 .
\label{E0-1.13}
\end{equation}
\item $\underbrace{v(y_1,y_2,z_3,0)}_{=1 \ \mathrm{by\ monotonicity}} = 1 + u_{3,4}(z_3,0) + u_{1,3}(y_1,z_3) + u_{2,3}(y_2,z_3)$.
Hence
\begin{equation} 
  u_{3,4}(z_3,0) + u_{1,3}(y_1,z_3) + u_{2,3}(y_2,z_3) = 0 .
\label{E0-1.14}
\end{equation}
\item $\underbrace{v(y_1,y_2,0,z_4)}_{=1 \ \mathrm{by\ monotonicity}} = 1 + u_{3,4}(0,z_4) + u_{1,4}(y_1,z_4) + u_{2,4}(y_2,z_4)$.
Hence
\begin{equation} 
  u_{3,4}(0,z_4) + u_{1,4}(y_1,z_4) + u_{2,4}(y_2,z_4) = 0 .
\label{E0-1.15}
\end{equation}
\item $\underbrace{v(y_1,0,z_3,z_4)}_{=1 \ \mathrm{by\ monotonicity}} = u_{1,2}(y_1,0) + 1 + u_{1,3}(y_1,z_3) + u_{1,4}(y_1,z_4)$.
Moreover, $u_{1,2}(y_1,0)=0$ by \refe{E0-1.9}.
Hence
\begin{equation} 
  u_{1,3}(y_1,z_3) + u_{1,4}(y_1,z_4) = 0 .
\label{E0-1.16}
\end{equation}
\item $\underbrace{v(0,y_2,z_3,z_4)}_{=1 \ \mathrm{by\ monotonicity}} = u_{1,2}(0,y_2) + 1 + u_{2,3}(y_2,z_3) + u_{2,4}(y_2,z_4)$.
Hence
\begin{equation} 
  u_{1,2}(0,y_2) + u_{2,3}(y_2,z_3) + u_{2,4}(y_2,z_4) = 0 .
\label{E0-1.17}
\end{equation}
\item From \refe{E0-1.16}, \refe{E0-1.17} and \refe{E0-1.13},
\begin{equation} 
  u_{1,2}(0,y_2) = 1 .
\label{E0-1.18}
\end{equation}
\item $v(0,y_2,z_3,0) = \underbrace{u_{1,2}(0,y_2)}_{=1 \ \mathrm{by\ \refe{E0-1.18}}} + u_{3,4}(z_3,0) + u_{2,3}(y_2,z_3)$.
Moreover, $v(0,y_2,z_3,0) \geq v(0,y_2,0,0) = u_{1,2}(0,y_2) = 1$.
Hence
\begin{equation} 
  u_{3,4}(z_3,0) + u_{2,3}(y_2,z_3) = 0  \mbox{ and } u_{2,3}(y_2,z_3) \in \{-1,0\} .
\label{E0-1.20}
\end{equation}
\item $v(0,y_2,0,z_4) = 1 + u_{3,4}(0,z_4) + u_{2,4}(y_2,z_4)$.
Moreover, $v(0,y_2,0,z_4) \geq v(0,y_2,0,0) = u_{1,2}(0,y_2) = 1$.
Hence
\begin{equation} 
  u_{3,4}(0,z_4) + u_{2,4}(y_2,z_4) = 0  \mbox{ and } u_{2,4}(y_2,z_4) \in \{-1,0\} .
\label{E0-1.23}
\end{equation}
\end{itemize}

\medskip

From \refe{E0-1.13} and \refe{E0-1.16}, we get
$u_{2,3}(y_2,z_3)+u_{2,4}(y_2,z_4)=-1$.
As $u_{2,3}(y_2,z_3),u_{2,4}(y_2,z_4) \in \{-1,0\}$ (by \refe{E0-1.20} and \refe{E0-1.23}), we have two cases:
\begin{itemize}
\item \underline{Case 1: $u_{2,3}(y_2,z_3) = -1$ and $u_{2,4}(y_2,z_4) = 0$}. 
Then
\begin{align*}
 & u_{3,4}(z_3,0) = 1 \quad  \mbox{by \refe{E0-1.20}} \\
 & u_{1,3}(y_1,z_3) = 0 \quad  \mbox{by \refe{E0-1.14}} \\ 
 & u_{1,4}(y_1,z_4) = 0 \quad  \mbox{by \refe{E0-1.16}} \\ 
 & u_{1,2}(y_1,0) = 0 \quad  \mbox{by \refe{E0-1.9}} \\ 
 & u_{1,2}(0,y_2) = 1 \quad  \mbox{by \refe{E0-1.18}} \\
 & u_{3,4}(0,z_4) = 0 \quad  \mbox{by \refe{E0-1.23}}
\end{align*}
All values are determined.
\item \underline{Case 2: $u_{2,3}(y_2,z_3) = 0$ and $u_{2,4}(y_2,z_4) = -1$}. 
Then
\begin{align*}
 & u_{3,4}(0,z_4) = 1 \quad  \mbox{by \refe{E0-1.23}} \\
 & u_{3,4}(z_3,0) = 0 \quad  \mbox{by \refe{E0-1.20}} \\
 & u_{1,3}(y_1,z_3) = 0 \quad  \mbox{by \refe{E0-1.14}} \\ 
 & u_{1,4}(y_1,z_4) = 0 \quad  \mbox{by \refe{E0-1.15}} \\ 
 & u_{1,2}(y_1,0) = 0 \quad  \mbox{by \refe{E0-1.9}} \\ 
 & u_{1,2}(0,y_2) = 1 \quad  \mbox{by \refe{E0-1.18}} 
\end{align*}
All values are determined.
\end{itemize}

\bigskip

\dep\underline{\bf Analysis with $y'$ and $z$:}
The analyses with $y$ and $z$, and with $y'$ and $z$ are similar.
By \refe{E0-1.9}, we just need to invert the two attributes $1$ and $2$.
Hence a similar reasoning to the previous analysis can be done.
We obtain thus the two cases $1'$ and $2'$ which are deduced from cases $1$ and $2$ just by switching attributes $1$ and $2$:

\begin{itemize}
\item \underline{Case 1':}
\begin{align*}
 & u_{1,3}(y'_1,z_3) = -1  \\
 & u_{1,4}(y'_1,z_4) = 0  \\ 
 & u_{3,4}(z_3,0) = 1  \\ 
 & u_{2,3}(y'_2,z_3) = 0  \\
 & u_{2,4}(y'_2,z_4) = 0  \\ 
 & u_{1,2}(y'_1,0) = 1   \\ 
 & u_{1,2}(0,y'_2) = 0  \\
 & u_{3,4}(0,z_4) = 0 
\end{align*}
\item \underline{Case 2':}
\begin{align*}
 & u_{1,3}(y'_1,z_3) = 0  \\ 
 & u_{1,4}(y'_1,z_4) = -1  \\ 
 & u_{3,4}(0,z_4) = 1 \\ 
 & u_{3,4}(z_3,0) = 0  \\
 & u_{2,3}(y'_2,z_3) = 0  \\ 
 & u_{2,4}(y'_2,z_4) = 0  \\
 & u_{1,2}(y'_1,0) = 1  \\ 
 & u_{1,2}(0,y'_2) = 0  
\end{align*}
\end{itemize}

\bigskip

\dep\underline{\bf Synthesis:}
Cases 1 and 2' are incompatible, and so are cases 2 and 1'.
We have thus the alternative:
\begin{itemize}
\item \underline{Case 1 and 1'}.
Gathering the values of partial utilities, we get
\[ \begin{array}{lllll}
    u_{1,2}(0,y_2)= 1 & \quad & u_{1,4}(y'_1,z_4) = 0 & \quad & u_{1,3}(y'_1,z_3) = -1\\
		u_{2,3}(y_2,z_3) = -1  & &  u_{2,4}(y_2,z_4) = 0
	 \end{array}
\]
As $u_{1,2}(y'_1,y_2) \geq u_{1,2}(0,y_2)=1$, we have $u_{1,2}(y'_1,y_2) = 1$.
Hence
\begin{align*}
  u(y'_1,y_2,z_3,z_4) & = \underbrace{u_{1,2}(y'_1,y_2)}_{=1} + \underbrace{u_{3,4}(z_3,z_4)}_{=1} 
  + \underbrace{u_{1,3}(y'_1,z_3)}_{=-1}  \\ 
	& + \underbrace{u_{1,4}(y'_1,z_4)}_{=0}  + \underbrace{u_{2,3}(y_2,z_3)}_{=-1}  + \underbrace{u_{2,4}(y_2,z_4)}_{=0} \\
	 & = 0
\end{align*}
We obtain a contradiction as $u(y'_1,y_2,z_3,z_4) \geq u(0,0,z_3,z_4) = 1$.

\item \underline{Case 2 and 2'}.
Gathering the values of partial utilities, we get
\[ \begin{array}{lllll}
    u_{1,2}(0,y_2) = 1 & \quad & u_{1,4}(y'_1,z_4) = -1  & \quad &  u_{1,3}(y'_1,z_3) = 0 \\
		u_{2,3}(y_2,z_3) = 0 & &  u_{2,4}(y_2,z_4) = -1
	 \end{array}
\]
As $u_{1,2}(y'_1,y_2) \geq u_{1,2}(0,y_2)=1$, we have $u_{1,2}(y'_1,y_2) = 1$.
Hence
\begin{align*}
  u(y'_1,y_2,z_3,z_4) & = \underbrace{u_{1,2}(y'_1,y_2)}_{=1} + \underbrace{u_{3,4}(z_3,z_4)}_{=1} 
  + \underbrace{u_{1,3}(y'_1,z_3)}_{=0}  \\ 
	& + \underbrace{u_{1,4}(y'_1,z_4)}_{=-1}  + \underbrace{u_{2,3}(y_2,z_3)}_{=0}  + \underbrace{u_{2,4}(y_2,z_4)}_{=-1} \\
	 & = 0
\end{align*}
We obtain a contradiction as $u(y'_1,y_2,z_3,z_4) \geq u(0,0,z_3,z_4) = 1$.
\end{itemize}
A contradiction is raised in all situations.
Hence it is not possible to have $u_{3,4}$ non-zero, knowing that $u_{1,2}$ depends on its two variables.
\end{proof}

\medskip

{\bf Step 3: Case where $u$ has two non-zero terms with overlapping support,
  e.g., $u_{1,2}$ and $u_{1,3}$.} In the last case, term $u_{1,2}$ depends on its
two variables, and there is no non-zero term $u_{i,j}$, with $i,j \not = 1,2$,
that depends on its two variables.

We proceed as in the beginning of Step~2, assuming that
\begin{equation}
  \forall i\in \{1,2\} \ \forall j\in N\setminus \{1,2\} \ \forall x_i\in X_i  \qquad  u_{i,j}(x_i,0)=0 .
\label{E0-1.101}
\end{equation}
Then relations \refe{E0-1.7} through \refe{E0-1.9} also hold in this case.

Assume by contradiction that there exists a non-zero $u_{i,j}$ for some $i\in\{1,2\}$ and $j \in N\setminus\{1,2\}$. 
Wlog, we assume it is $u_{1,3}$.
There exists thus $z_1 \in X_1$ and $z_3 \in X_3$ such that 
\begin{equation}
 u_{1,3}(z_1,z_3) \not=0 .
\label{E0-1.102}
\end{equation}
One can transfer term $u_{i,3}(0,x_3)$, for $i\not= 1,3$, to $u_{1,3}$ (proceeding as in the beginning of Step~2). Hence we can assume that
\begin{equation}
 \forall i\in N\setminus \{1,3\} \ \forall x_3\in X_3 \qquad u_{i,3}(0,x_3) = 0 .
\label{E0-1.101bis}
\end{equation}
\begin{lemma}
If $u_{1,2}$ depends on its two
variables, then $u_{1,3}$ is identically zero.
\label{L0-1.3}
\end{lemma}
\begin{proof}
We set $v(x_1,x_2,x_3) = u(x_1,x_2,x_3,0_{-1,2,3})$. Then
\[ v(x_1,x_2,x_3) = u_{1,2}(x_1,x_2) + u_{1,3}(x_1,x_3) + u_{2,3}(x_2,x_3) .
\]

\dep\underline{\bf Analysis with $y$ and $z$:}
We write thanks to \refe{E0-1.7} and to the monotonicity of $v$
\begin{align*}
 & v(z_1,0,z_3) = u_{1,2}(z_1,0) + u_{1,3}(z_1,z_3) \\
 & \geq v(z_1,0,0) = u_{1,2}(z_1,0)
\end{align*}
Hence $u_{1,3}(z_1,z_3) \geq 0$, which gives by \refe{E0-1.102}
\begin{align}
 & u_{1,3}(z_1,z_3) = 1  \label{E0-1.116} \\
 & u_{1,2}(z_1,0) = 0  \label{E0-1.117}
\end{align}
We have the following basic relations:
\begin{align}
 & v(y_1,0,z_3) = \underbrace{u_{1,2}(y_1,0)}_{=0} + u_{1,3}(y_1,z_3)
  \label{E0-1.109} \\
 & v(y_1,y_2,z_3) = 1 + u_{1,3}(y_1,z_3) + u_{2,3}(y_2,z_3)
  \label{E0-1.110} \\
 & v(z_1,y_2,z_3) = u_{1,2}(z_1,y_2) + u_{1,3}(z_1,z_3) + u_{2,3}(y_2,z_3) 
  \label{E0-1.111} 
\end{align}

\bigskip

\dep\underline{\bf Analysis with compound alternatives:}
We distinguish between two cases:
\begin{itemize}
\item Assume first that $z_1 \geq y_1$.
By \refe{E0-1.8} and \refe{E0-1.9}, we have
\begin{equation}
  u_{1,2}(z_1,y_2)=1 .
	\label{E0-1.122}
\end{equation}
By monotonicity, $v(z_1,y_2,z_3)=1$ (as $v(z_1,0,z_3)=u_{1,2}(z_1,0)+1$ and thus $v(z_1,0,z_3)=1$). 
Hence \refe{E0-1.116} and \refe{E0-1.111} give
\begin{equation}
  u_{2,3}(y_2,z_3) = -1 .
\label{E0-1.120}
\end{equation}
By monotonicity, $v(y_1,y_2,z_3)=1$ (as $v(y_1,y_2,0)=u_{1,2}(y_1,y_2)=1$). 
From \refe{E0-1.110} and previous relation, we have
\begin{equation}
  u_{1,3}(y_1,z_3) = 1 .
\label{E0-1.121}
\end{equation}
\item Assume then that $z_1 < y_1$.
We have $v(y_1,0,z_3)=1$ by monotonicity of $v$ (as $v(z_1,0,z_3)=1$).
Then \refe{E0-1.109} proves that \refe{E0-1.121} holds.
This implies that \refe{E0-1.120} also holds, thanks to \refe{E0-1.110}.

By monotonicity, $v(z_1,y_2,z_3)=1$ (as $v(z_1,0,z_3)=1$). 
Hence \refe{E0-1.111}  and \refe{E0-1.120} show that \refe{E0-1.122} is satisfied.
\end{itemize}
In the two cases, we have proved that relations \refe{E0-1.122}, \refe{E0-1.120} and \refe{E0-1.121} are true.

\medskip

We make the following reasoning.
\begin{itemize}
\item We write
\begin{align*}
 & v(0,y_2,z_3) = u_{1,2}(0,y_2) + u_{1,3}(0,z_3) - 1 \\
 & \geq v(0,y_2,0) = u_{1,2}(0,y_2)
\end{align*}
Therefore $u_{1,3}(0,z_3) \geq 1$. 
We also see that $u_{1,3}(0,z_3) \in \{0,1\}$ as $v(0,0,z_3) = u_{1,3}(0,z_3)$.
Hence
\begin{align}
  & u_{1,3}(0,z_3) = 1 \label{E0-1.123} \\
	& v(0,0,z_3) = 1 \label{E0-1.125}
\end{align}
\item We write
\begin{align*}
 & v(0,y_2,z_3) = u_{1,2}(0,y_2) + u_{1,3}(0,z_3) + u_{2,3}(y_2,z_3) = u_{1,2}(0,y_2) \\
 & \geq v(0,0,z_3) = 1
\end{align*}
Hence
\begin{equation}
  u_{1,2}(0,y_2) = 1 .
\label{E0-1.124} 
\end{equation}
\item We have
\begin{align*}
 & \underbrace{v(0,y'_2,z_3)}_{=1 \ \mbox{\scriptsize by monotonicity and \refe{E0-1.125}}} = 
  \underbrace{u_{1,3}(0,z_3)}_{=1} + u_{2,3}(y'_2,z_3) 
\end{align*}
Hence
\begin{equation}
  u_{2,3}(y'_2,z_3) = 0 .
\label{E0-1.127} 
\end{equation}
\item We have
\begin{align*}
 & \underbrace{v(y'_1,y'_2,z_3)}_{=1 \ \mbox{\scriptsize by monotonicity}} = 
  1 + u_{1,3}(y'_1,z_3) + \underbrace{u_{2,3}(y'_2,z_3)}_{=0 \ \mbox{by \refe{E0-1.127}}} 
\end{align*}
Hence
\begin{equation}
  u_{1,3}(y'_1,z_3) = 0 .
\label{E0-1.128} 
\end{equation}
\item Finally
\begin{align*}
 & v(y'_1,y_2,z_3) = 
  \underbrace{u_{1,2}(y'_1,y_2)}_{=1 \ \mbox{\scriptsize by \refe{E0-1.8} and \refe{E0-1.124}}} 
	+ \underbrace{u_{1,3}(y'_1,z_3)}_{=0 \ \mbox{by \refe{E0-1.128}}} 
	+ \underbrace{u_{2,3}(y_2,z_3)}_{=-1 \ \mbox{by \refe{E0-1.120}}} 
	= 0
\end{align*}
We obtain a contradiction as $v(y'_1,y_2,z_3)=1$ (thanks to monotonicity of $v$, and to \refe{E0-1.125}).
\end{itemize}
A contradiction is raised in all situations.
Hence it is not possible to have $u_{1,3}$ non-zero, knowing that $u_{1,2}$ depends on its two variables.
\end{proof}
Finally, we have proved that if $u_{1,2}$ depends on its two variables, no other term can be non-zero.
This proves that $u$ depends only on two variables.
\end{proof}
In summary, we have proved that the extreme points of $\cP_{k,2}$ are the
2-additive 0-1-valued $k$-ary capacities, and that these capacities have a
support of size at most 2. It follows that any $v\in\cP_{k,2}$ can be written as
a convex combination of 2-additive $k$-ary capacities with support of size at
most 2, which proves Theorem~\ref{th:main}.

\subsection{Expression of the extreme points of the polytope of 2-additive
  $k$-ary capacities} 
We are now in position to determine all vertices of $\cP_{k,2}$, for a fixed $k\in
\NN$. By Theorem~\ref{th:support}, we know that any vertex has a support of at most
two elements, hence w.l.o.g. we can restrict to elements 1 and 2. By
Theorem~\ref{th:2}, finding all
vertices with support $\{1,2\}$ amounts to finding all 0-1 $k$-ary capacities
which are linear combinations of unanimity games $u_x$ with $\supp(x)\subseteq
\{1,2\}$. By analogy with classical simple games, a coalition $x\in L^N$ is
\textit{winning} for $v$ if $v(x)=1$. Minimal winning coalitions are those which
are minimal w.r.t. the order $\leq$ on $L^N$, and therefore they form an
antichain in $L^N$. We show several properties of minimal winning coalitions.
\begin{lemma}\label{lem:simp}
Let $\mu$ be a 0-1-valued $k$-ary capacity.
\begin{enumerate}
\item $x$ is a minimal winning coalition if and only if $m^\mu(x)=1$ and $m^\mu(y)=0$
  for all $y< x$. 
\item $\supp(\mu)\subseteq \{1,2\}$ if and only if its
  minimal winning coalitions have support included in $\{1,2\}$.
\item If $|\supp(\mu)|=2$, there are at most $k+1$ distinct minimal
  winning coalitions.
\item Suppose that $\supp(\mu)\subseteq\{1,2\}$. Denote by
  $x^1,\ldots, x^q$ the minimal winning coalitions of $\mu$, arranged such that
  $x_1^1<x_1^2\cdots<x_1^q$. Then $m^\mu(x^\ell)=1$ for all $\ell=1,\ldots, q$,
  $m^\mu(x^\ell\vee x^{\ell+1})=-1$ for $\ell=1,\ldots, q-1$, and $m^\mu(x)=0$ otherwise. 
\end{enumerate}
\end{lemma}
\begin{proof}
\begin{enumerate}
\item Suppose $m^\mu(x)=1$ and $m^\mu(y)=0$ for all $y< x$. Then clearly $x$
  is a minimal winning coalition. Conversely, suppose first that there exists $y<x$ such
  that $m^\mu(y)\neq 0$, and choose a minimal $y$ with this property. Then
  $\mu(y)\neq 0$, a contradiction. Then, suppose there is no such $y< x$
  but $m^\mu(x)\neq 1$. Then $\mu(x)=m^\mu(x)\neq 1$, again a contradiction.
\item Suppose there exists a minimal winning coalition $x$ such that
  $\supp(x)\not\subseteq \{1,2\}$. Then by (i), the support of $\mu$ is
  not included in $\{1,2\}$.

Conversely, suppose that there exists $x\in L^N$ with $m^\mu(x)\neq 0$ and
$\supp(x)\not\subseteq \{1,2\}$. Choose a minimal such $x$. By
Lemma~\ref{lem:3}, $m^\mu(x)\in\{-1,0,1\}$. Observe that
$m^\mu(x)=-1$ is impossible, because this would yield $\mu(x)=-1$. Then
$m^\mu(x)=1=\mu(x)$, proving by (i) that $x$ is a minimal winning
coalition.
\item Take $x$ being a minimal winning coalition, and suppose w.l.o.g. that
  $\supp(x)\subseteq \{1,2\}$. Observe that any other minimal winning
  coalition $y$ must satisfy $x_1\neq y_1$, otherwise one of the two would not be
  minimal. Hence, there can be at most $k+1$ distinct minimal winning
  coalitions. 
\item By uniqueness of the decomposition, it suffices to check that the
  computation of $\mu$ by $\mu(x)=\sum_{y\leq x}m^\mu(y)$ works. By
  construction, any $x\in L^N$ is greater or equal to a subset of consecutive
  minimal winning coalitions, say, $x^{i+1},x^{i+2},\ldots, x^{i+j}$, so that there
  are $j-1$ pairs $(x^{i+\ell},x^{i+\ell+1})$, $\ell=1,\ldots,j-1$. The result
  follows by the definition of $m^\mu$.    
\end{enumerate}
\end{proof}
The various properties in the Lemma permit to say that the vertices of $\cP_{k,2}$
with support included into $\{1,2\}$ are in bijection with the antichains (which
are of size at most $k+1$) in the
lattice $(k+1)^2$. Moreover, their M\"obius transform is known.

\begin{lemma}\label{lem:anti}
Let $k\in \NN$. Denote by $\kappa(\ell)$ the number of antichains of $\ell$ elements
in the lattice $(k+1)^2$, $\ell=1,\ldots, k+1$. Then
\[
\kappa(\ell) = \binom{k+1}{\ell}^2.
\]
Moreover, the total number of antichains on $(k+1)^2$ is
\[
\sum_{\ell=1}^{k+1}\kappa(\ell) = \binom{2k+2}{k+1}-1.
\]
\end{lemma}
\begin{proof}
Let $x\in(k+1)^2$, with coordinates $(x_1,x_2)$. Considering that the 1st
coordinate axis is on the left, we say that $y$ is on the left of $x$ if
$y_1>x_1$ and $y_2<x_2$. Let us denote by $F_1(x_1,x_2)$
the number of points $y$ to the left of $x$ (i.e., $\{x,y\}$ is an antichain). We obtain
\[
F_1(x_1,x_2) = \sum_{y_2=0}^{x_2-1}\sum_{y_1=x_1+1}^k 1 = x_2(k-x_1).
\]
Note that $\kappa(1) = F_1(-1,k+1)$ since any point in $(k+1)^2$ is to the left
of $(-1,k+1)$.

Define $F_2(x_1,x_2)$ as the number of antichains $\{y,z\}$ to the left of $x$,
with $z$ to the left of $y$, 
i.e., $\{x,y,z\}$ forms an antichain. We obtain
\[
F_2(x_1,x_2) = \sum_{y_2=1}^{x_2-1}\sum_{y_1=x_1+1}^{k-1}F_1(y_1,y_2).
\]
(note that $y_2=0$ and $y_1=k$ are impossible because $z$ is on the left of $y$). Again
remark that $\kappa(2) = F_2(-1,k+1)$. More generally, the number of antichains
of $\ell$ elements to the left of $x$ is
\[
F_\ell(x_1,x_2) =
\sum_{y_2=\ell-1}^{x_2-1}\sum_{y_1=x_1+1}^{k-\ell+1}F_{\ell-1}(y_1,y_2) \qquad
(1\leq\ell\leq k+1),
\]
and $\kappa(\ell) = F_\ell(-1,k+1)$. We show by induction that 
\begin{equation}\label{eq:i1}
F_\ell(x_1,x_2) = \binom{x_2}{\ell}\binom{k-x_1}{\ell}.
\end{equation}
The result has already been verified for $\ell=1$. We assume it is true up to some
integer $1\leq\ell\leq k$ and prove it for $\ell+1$. We have
\begin{align*}
F_{\ell+1}(x_1,x_2) & =
\sum_{y_2=\ell}^{x_2-1}\sum_{y_1=x_1+1}^{k-\ell}F_{\ell}(y_1,y_2)\\
 & = \sum_{y_2=\ell}^{x_2-1}\sum_{y_1=x_1+1}^{k-\ell}
\binom{y_2}{\ell}\binom{k-y_1}{\ell}\\
 & =
\sum_{y_2=\ell}^{x_2-1}\binom{y_2}{\ell}\sum_{y_1=x_1+1}^{k-\ell}\binom{k-y_1}{\ell}\\
 & = \binom{x_2}{\ell+1}\binom{k-x_1}{\ell+1},
\end{align*}
where we have used the fact that (see \cite[\S 0.151]{grry07})
\[
\sum_{k=0}^m\binom{n+k}{n} = \binom{n+m+1}{n+1}.
\]
Hence (\ref{eq:i1}) is proved. It remains to compute the total number of
antichains. Using the fact that (see \cite[\S 0.157]{grry07})
\[
\sum_{k=0}^n\binom{n}{k}^2 = \binom{2n}{n},
\]
we find the desired result.
\end{proof}
Observing that the antichain $\{\0\}$ does not correspond to a normalized
capacity, we obtain directly from Lemma~\ref{lem:anti} and previous
considerations the following result.
\begin{theorem}
Let $k\in \NN$ and consider the polytope $\cP_{k,2}$. The following holds.
\begin{enumerate}
\item For any $i\in N$, the number of vertices with support $\{i\}$ is
  $k$.
\item For any distinct $i,j\in N$, the number of vertices with support included
  in $\{i,j\}$ is $\displaystyle \binom{2k+2}{k+1}-2$.
\item The total number of vertices of $\cP_{k,2}$ is
\[
  \Bigg[\binom{2k+2}{k+1}-2 - 2k\Bigg]\frac{n(n-1)}{2} + kn =
  \Bigg[\binom{2k+2}{k+1}-2\Bigg]\frac{n(n-1)}{2} - kn(n-2).
\]
\end{enumerate}
\end{theorem}
 
\subsection{Significance of the main theorem}

We have seen in Section \ref{S2.3} that the decomposition of a GAI model is
not unique in general, and moreover, nothing ensures that the terms of the
decomposition have the same type of monotonicity (see  Example~\ref{Ex1}). 

According to Theorem \ref{th:main}, any monotone $2$-additive discrete GAI model
can be rewritten using only nonnegative and monotone utility terms, which is not
the case of the canonical decomposition (see Example \ref{Ex1}). Hence, using
our decomposition, it is easy to provide a semantics to each utility terms
$u_S$, so that the model can be easily interpreted and displayed to the decision maker.

\medskip

Theorem \ref{th:main} brings also very important benefits during the elicitation
of a GAI model.  It indeed reduces the representation of monotonicity
constraints from exponential to quadratic complexity.  The aim of elicitation is
to construct the parameters of the decision model from preference information.
Classically, preference information consists of a set of pairwise comparisons
among elements in $X$ (for each pair $(x,y)\in X^2$, the decision maker strictly
prefers $x$ to $y$) or an assignment of elements in $X$ to some predefined
ordered categories $C_1,\ldots,C_m$ as in classification problems.  There are
mainly two elicitation paradigms.

The first one is a constraint approach, where each pair $(x,y)$ is transformed
into a linear constraint on the parameters of the GAI model
\citep{grmosl12,bifameza12,lagr13}.  Monotonicity conditions can also be
written as linear constraints.  The GAI model is then identified using Linear
Programming.   The practical identification of the model appears to be rapidly
computationally intractable as the number of attributes and the cardinality of
the attributes grow.  Thanks to Theorem \ref{th:main}, these constraints can be
replaced by monotonicity conditions on each term $u_S$ in the GAI decomposition,
which reduces the number of monotonicity constraints from exponential to quadratic
in the number of criteria.  This permits to solve problems of much larger
size.

Within a constraint approach, robust methods are appealing as they consider all
parameters values fulfilling the previous constraints, rather than arbitrarily
selecting one of these values.  MinMax Regret criterion is a conservative way to
handle the uncertainty on the decision model \citep{bopaposc06}.  The idea is
to set bounds on the worst possible loss one could have by choosing an
alternative, looking at the set of possible parameters values.  It is
interesting to note that the scientific community that developped these
approaches does not enforce monotonicity conditions.  This makes the elicitation
quite complex, as one needs to provide a lot of preference information to obtain
the correct monotonicity conditions.  Most applications in this area consider a
very small size of $\mathcal{S}$ compared to the number of criteria, which is
not always possible in practice.  One would then expect a great benefit of enforcing
monotonicity conditions in the MinMax Regret method.  Here again, Theorem
\ref{th:main} is very helpful as it reduces the number of monotonicity
conditions to a tractable number.

Methods of the second paradigm are statistical.  One can mention as an example
the extension of Logistic Regression to utility models incorporating interaction
among criteria \citep{fachdehu12,falahu14}. Here the preference information
is put into the function to optimize, and the global problem to solve is often a
convex problem under linear constraints, mostly monotonicity conditions. The
number of monotonicity conditions highly influences the efficiency of the
optimization algorithm. This shows again the importance of Theorem
\ref{th:main}.

\section{Conclusion}
We have shown in this paper that it is always possible to write a 2-additive
discrete GAI model as a sum of nonnegative and monotone nondecreasing terms,
thus reducing the complexity of any optimization problem involving such models
from exponential to quadratic complexity in the number of attributes. We believe
that this result opens the way to the practical utilization of GAI models. 

By the equivalence between 2-additive discrete GAI models and 2-additive $k$-ary
capacities, as a by-product of our main result, we have obtained all extreme
points of the polytope of 2-additive $k$-ary capacities, a result which is new,
as far as we know, and which generalizes the results of \cite{micogi06} for
classical 2-additive capacities.

\appendix

\section{M\"obius transform of a $k$-ary capacity}
The result can be easily obtained by using standard results of the theory of
M\"obius functions (see, e.g., \cite{aig79}). Given a finite poset (partially ordered
set) $(P,\leq)$, its \textit{M\"obius function} $\mu:P\times P\rightarrow \RR$
is defined inductively by:
\[
\mu(x,y) = \begin{cases}
  1, & \text{if } x=y\\
  -\sum_{x\leq t<y}\mu(x,t), & \text{if }x<y\\
  0, & \text{otherwise}
  \end{cases}.
\]
Then the solution of the system $f(x)=\sum_{y\leq x}g(y), x\in P$, is given by
\[
g(x) = \sum_{y\leq x}\mu(y,x)f(y) \qquad (x\in P),
\]
and $g$ is called the M\"obius transform (or inverse) of $f$. Note that in the
case of capacities, $(P,\leq)$ is taken as $(2^N,\subseteq )$.

Considering two posets $(P,\leq)$, $(P',\leq')$, and the product poset $(P\times
P',\leq)$ where $\leq$ is the product order, i.e., $(x,y)\leq (x',y')$ if $x\leq
x'$ and $y\leq'y'$, it is easy to show that the M\"obius function on $P\times
P'$ is the product of the M\"obius functions on $P$ and $P'$:
\[
\mu((x,t),(y,z)) = \mu_P(x,y)\mu_{P'}(t,z) \quad (x,y\in P, t,z\in P').
\]
Let us apply this result to $k$-ary capacities. It is easy to see that the
M\"obius function on the chain $\{0,1,\ldots,k\}$ is given by
\begin{equation}\label{eq:mobkf}
\mu_{\{0,1,\ldots,k\}}(x,y) = \begin{cases}
  (-1)^{y-x}, & \text{if } 0\leq y-x\leq 1\\
  0, & \text{otherwise.}
  \end{cases}
\end{equation}
It follows that the M\"obius transform $m^v$ of a $k$-ary capacity $v$ is given
by
\[
m^v(x) = \sum_{y\leq x:x_i-y_i\leq 1 \forall i\in N}(-1)^{\sum_{i\in N}(x_i - y_i)}v(y).
\]

\section{Acknowledgments}
The corresponding author thanks the Agence Nationale de la Recherche for
financial support under contract ANR-13-BSHS1-0010 (DynaMITE).

\bibliographystyle{plainnat}
\bibliography{../BIB/fuzzy,../BIB/grabisch,../BIB/general}

\end{document}